\documentclass[journal]{IEEEtran}
\usepackage{amsmath,amssymb,amsthm}
\usepackage{microtype}

\newtheorem{theorem}{Theorem}
\newtheorem{lemma}{Lemma}
\newtheorem{proposition}{Proposition}
\newtheorem{corollary}{Corollary}
\theoremstyle{definition}
\newtheorem{definition}{Definition}
\newtheorem{example}{Example}
\theoremstyle{remark}
\newtheorem{remark}{Remark}

\newcommand{\F}{\mathbb{F}}
\newcommand{\R}{\mathbb{R}}
\newcommand{\N}{\mathbb{N}}

\begin{document}

\title{Trellis State Complexity as an Exact\\ Tropical Factorization Rank}

\author{Karthik~Sheshadri%
\thanks{K.~Sheshadri (e-mail: karthiksheshadri217@gmail.com).}}

\markboth{Submitted to the IEEE Transactions on Information Theory}%
{Sheshadri: Trellis State Complexity as an Exact Tropical Factorization Rank}

\maketitle

\begin{abstract}
Let $C\subseteq\F_2^m$ be a binary linear code and let $[m]=L\sqcup R$ be a
bipartition of its coordinates. The \emph{conditional decoding matrix} of $C$
at this cut is the matrix $W$ indexed by $\F_2^{L}\times\F_2^{R}$ whose entry
$W(x_L,x_R)$ is the coset-leader weight $d\bigl((x_L,x_R),C\bigr)$, the
minimum Hamming distance from the word $(x_L,x_R)$ to the code. We prove that
the min-plus factorization rank (Barvinok rank) of $W$, and likewise its
tropical rank, equal $2^{s}$ exactly, where $s=\dim C-\dim C_L-\dim C_R$ is
the classical state complexity of the minimal trellis of $C$ at the cut. The
upper bound is a two-party reading of Viterbi decoding on the minimal
trellis; the contribution is the matching lower bound, which holds against
arbitrary min-plus factorizations rather than only sequential trellis
realizations, and is obtained from an explicit $2^{s}\times 2^{s}$ tropically
nonsingular submatrix built from a transversal of codewords. Specializing $C$
to the cut space of a graph identifies $W$ with the conditional ground-state
energy of Ising signings (the frustration index), and yields natural graph
families whose conditional matrices have min-plus rank exponential in the
number of vertices; for these families we also record the contrasting local
statement that all bounded-radius views of a signing are switching-trivial,
so the exponential rank is carried entirely by non-local structure. We note
explicitly that this rank measures representational incompressibility, not
computational hardness: planar families attain the same exponential rank
while their ground states are computable in polynomial time.
\end{abstract}

\begin{IEEEkeywords}
Trellis complexity, minimal trellis, state complexity, Viterbi decoding,
tropical semiring, min-plus algebra, Barvinok rank, coset leader, signed
graphs, Ising model, frustration index.
\end{IEEEkeywords}

\IEEEpeerreviewmaketitle

\section{Introduction}

\IEEEPARstart{T}{wo} classical bodies of theory measure, in different
vocabularies, how much information must cross a cut when one optimizes over a
binary linear code.

On the coding-theory side, the \emph{minimal trellis} of a code
$C\subseteq\F_2^m$ realizes $C$ as the path set of a layered graph, and the
size of the state space at a given cut is governed by the
\emph{state complexity}
\[
  s \;=\; \dim C-\dim C_L-\dim C_R ,
\]
where $C_L$ and $C_R$ are the subcodes supported on the two sides
(Muder~\cite{Muder88}; Forney~\cite{Forney94}; Kasami, Takata, Fujiwara and
Lin~\cite{KTFL93}; see Vardy~\cite{Vardy98} for a survey). Viterbi decoding
propagates $2^{s}$ conditional values across the cut, and among
\emph{trellis} realizations this is optimal.

On the tropical side, the \emph{min-plus factorization rank} (Barvinok rank,
factor rank) of a matrix $W$ is the least $r$ such that
$W(x,y)=\min_{t\le r}\bigl(u_t(x)+v_t(y)\bigr)$; together with the tropical
rank and the Kapranov rank it is one of the standard, mutually inequivalent
rank notions of tropical linear algebra introduced and compared by Develin,
Santos and Sturmfels~\cite{DSS05}. Computing the Barvinok rank is NP-hard
already for fixed target rank (Shitov~\cite{Shitov14}), so explicit natural
families with exactly known rank are of independent interest. The Barvinok
rank also has direct algorithmic content for optimization: distance matrices
of bounded factor rank admit polynomial-time travelling-salesman algorithms
(``TSP with warehouses''; see the discussion and references
in~\cite{Shitov14}).

This paper observes that the two theories meet in an exact identity. Attach
to the pair (code, cut) the \emph{conditional decoding matrix}
$W(x_L,x_R)=d\bigl((x_L,x_R),C\bigr)$: the table of optimal values of
minimum-distance decoding, doubly indexed by the two halves of the received
word. Viterbi decoding through the minimal trellis exhibits a min-plus
factorization of $W$ with $2^{s}$ terms. Our main result
(Theorem~\ref{thm:main}) shows that no factorization whatsoever does better:
for every binary linear code and every coordinate bipartition, the min-plus
factorization rank and the tropical rank of $W$ are both exactly $2^{s}$.

The lower bound is proved against arbitrary min-plus factorizations---there
is no sequentiality, layering, or state-machine assumption---by exhibiting an
explicit $2^{s}\times 2^{s}$ submatrix with zero diagonal and positive
off-diagonal entries, which is tropically nonsingular, together with a
two-line crossing argument showing that each rank-one term of any
factorization can be tight on at most one diagonal cell. Thus the minimal
trellis is optimal in a model considerably wider than trellises: the state
space is not merely the smallest \emph{sequential} realization of
conditional decoding, it is the smallest possible min-plus data structure of
any shape.

Specializing $C$ to the cut (cocycle) space of a connected graph $G$
identifies the coset-leader weight with the \emph{frustration index} of a
signing of $G$---equivalently, the ground-state energy of an Ising system
with $\pm1$ couplings---and $W$ with the conditional ground-state matrix
across an edge bipartition. For a natural class of bipartitions
(Corollary~\ref{cor:spoked}) the state complexity evaluates to $|B|-1$, where
$B$ is the ``far side'' of the vertex set, so the conditional matrix has
min-plus rank $2^{|B|-1}$, exponential in the number of vertices. We verify
the construction computationally on the Petersen graph and further random
instances (Remark~\ref{rem:verify}).

Two accompanying observations delimit what the theorem does and does not
say. First (Proposition~\ref{prop:local}), on graphs of girth $g$ every
radius-$r$ view of a signing with $2r+1<g$ is switching-equivalent to the
all-positive signing, so \emph{every} bounded-radius local statistic is
constant across all $2^{|E|}$ signings while the conditional matrix attains
the maximum rank permitted by its state complexity: the rank is carried
entirely by non-local structure. Second (Remark~\ref{rem:planar}),
exponential conditional rank is compatible with polynomial-time solvability:
planar prism graphs attain the bound of Corollary~\ref{cor:spoked} while
planar Ising ground states are computable in polynomial
time~\cite{Barahona82}. The conditional rank measures the incompressibility
of conditional \emph{values} in the min-plus model; it is not a hardness
measure, and we state this explicitly to forestall over-reading.

\section{Preliminaries}

\subsection{Codes and Cuts}
Throughout, $C\subseteq\F_2^m$ is a linear code and $[m]=L\sqcup R$ is a
fixed bipartition of the coordinates. For $x\in\F_2^m$ we write $x=(x_L,x_R)$
for the two restrictions and $|x|$ for the Hamming weight. The
\emph{coset-leader weight} of $x$ is
\[
  w_C(x)\;=\;d(x,C)\;=\;\min_{c\in C}\,|x\oplus c| .
\]
The \emph{past and future subcodes} at the cut are
\[
  C_L=\{c\in C: c_R=0\},\qquad C_R=\{c\in C: c_L=0\},
\]
which we freely identify with their supports in $\F_2^{L}$ and $\F_2^{R}$.
Write $P_R(C)=\{c_R:c\in C\}\subseteq\F_2^{R}$ for the projection; the
kernel of $c\mapsto c_R$ on $C$ is $C_L$, so
$\dim P_R(C)=\dim C-\dim C_L$, and symmetrically for $P_L$. The
\emph{state complexity} of $C$ at the cut is
\begin{align*}
  s\;=\;s(C;L,R)&=\dim C-\dim C_L-\dim C_R\\
   &=\dim P_R(C)-\dim C_R\\
   &=\dim P_L(C)-\dim C_L .
\end{align*}
This is the dimension of the state space of the minimal trellis of $C$ at
this cut position~\cite{Muder88,Forney94,KTFL93,Vardy98}; we will not need
any trellis machinery, only the dimension formula above, which we take as
the definition of $s$.

\subsection{The Conditional Decoding Matrix}
\begin{definition}
The \emph{conditional decoding matrix} of $C$ at the cut $(L,R)$ is the
matrix $W\in\N^{\F_2^{L}\times\F_2^{R}}$ with entries
\[
  W(x_L,x_R)\;=\;w_C\bigl((x_L,x_R)\bigr).
\]
\end{definition}
Every entry of $W$ is finite (indeed at most $m$).

\subsection{Tropical Ranks}
A matrix $M$ has \emph{min-plus factorization rank} (also \emph{Barvinok
rank} or \emph{factor rank}~\cite{DSS05,Shitov14}) at most $r$ if there are
functions $u_1,\dots,u_r$ on its row index set and $v_1,\dots,v_r$ on its
column index set, with values in $\R\cup\{+\infty\}$, such that
\[
  M(x,y)\;=\;\min_{1\le t\le r}\,\bigl(u_t(x)+v_t(y)\bigr)
  \qquad\text{for all }x,y .
\]
Allowing $+\infty$ only strengthens our lower bound. A square matrix $N$ of
order $r$ is \emph{tropically nonsingular} if the minimum in the tropical
determinant $\min_{\sigma\in S_r}\sum_i N(i,\sigma(i))$ is attained by a
unique permutation, and the \emph{tropical rank} of $M$ is the largest order
of a tropically nonsingular square submatrix~\cite{DSS05}. By~\cite{DSS05}
the tropical rank is at most the Kapranov rank, which is at most the
Barvinok rank; we will pin both ends directly and use this chain only for
the intermediate notion.

\section{The Door Identity: an Upper Bound From the State Space}

The state space at the cut is the quotient $T=P_R(C)/C_R$, of size $2^{s}$.
The upper bound is the observation that conditional decoding factors through
$T$; it is a two-party reading of Viterbi decoding, and we include the short
proof to keep the paper self-contained.

\begin{lemma}[fiber invariance]\label{lem:fiber}
For $r\in P_R(C)$ let $F(r)=\{c_L : c\in C,\ c_R=r\}\subseteq\F_2^{L}$.
Then $F(r\oplus k)=F(r)$ for every $k\in C_R$. Consequently $F$ and the
function $D(x_L,r)=d\bigl(x_L,F(r)\bigr)$ depend only on the class
$\tau=r\oplus C_R\in T$.
\end{lemma}

\begin{IEEEproof}
If $\kappa\in C$ has $\kappa_L=0$ and $\kappa_R=k$, then
$c\mapsto c\oplus\kappa$ is a bijection between $\{c\in C:c_R=r\}$ and
$\{c\in C:c_R=r\oplus k\}$ preserving left parts.
\end{IEEEproof}

\begin{proposition}[door identity]\label{prop:doors}
For every $x_L\in\F_2^{L}$ and $x_R\in\F_2^{R}$,
\[
  W(x_L,x_R)\;=\;\min_{\tau\in T}\;\Bigl(\,D(x_L,\tau)\;+\;d(x_R,\tau)\Bigr),
\]
where $d(x_R,\tau)=\min_{r\in\tau}|x_R\oplus r|$ is the distance from $x_R$
to the coset $\tau\subseteq\F_2^{R}$. In particular the min-plus
factorization rank of $W$ is at most $|T|=2^{s}$.
\end{proposition}

\begin{IEEEproof}
Group the defining minimum
$W(x)=\min_{c\in C}(|x_L\oplus c_L|+|x_R\oplus c_R|)$ by the right part
$r=c_R\in P_R(C)$: the inner minimum over left parts is $D(x_L,r)$, giving
$W(x)=\min_{r\in P_R(C)}\bigl(D(x_L,r)+|x_R\oplus r|\bigr)$. By
Lemma~\ref{lem:fiber}, $D(x_L,\cdot)$ is constant on classes, so collecting
each class and minimizing $|x_R\oplus r|$ within it yields the identity.
Each class contributes one rank-one term $u_\tau(x_L)+v_\tau(x_R)$ with
$u_\tau=D(\cdot,\tau)$ and $v_\tau=d(\cdot,\tau)$.
\end{IEEEproof}

\section{The Main Theorem: a Matching Lower Bound}

\begin{theorem}\label{thm:main}
For every binary linear code $C\subseteq\F_2^m$ and every bipartition
$[m]=L\sqcup R$, the min-plus factorization rank of the conditional
decoding matrix $W$ equals $2^{s}$, where $s=\dim C-\dim C_L-\dim C_R$.
Moreover the tropical rank and Kapranov rank of $W$ also equal $2^{s}$.
\end{theorem}

The proof rests on an explicit block and a crossing argument.

\begin{lemma}[transversal block]\label{lem:block}
Let $r_1,\dots,r_N\in P_R(C)$ be a transversal of $T=P_R(C)/C_R$, so
$N=2^{s}$, and for each $i$ choose $c^i\in C$ with $c^i_R=r_i$. Put
$x^i_L=c^i_L$ and $x^i_R=r_i$. Then
\[
  W(x^i_L,x^i_R)=0\ \ \text{for all }i,\qquad
  W(x^i_L,x^j_R)\ge 1\ \ \text{for all }i\ne j .
\]
\end{lemma}

\begin{IEEEproof}
The diagonal entry is $w_C(c^i)=0$. For $i\ne j$,
$(x^i_L,x^j_R)=c^i\oplus(0,\,r_i\oplus r_j)$. If this word were in $C$,
then $(0,\,r_i\oplus r_j)\in C$, hence $(0,r_i\oplus r_j)\in C_R$, hence
$r_i\oplus r_j\in C_R$, contradicting that $r_i,r_j$ lie in distinct
classes of the transversal. So the word is not a codeword and its
coset-leader weight is at least $1$.
\end{IEEEproof}

\begin{lemma}[crossing]\label{lem:crossing}
Let $N\times N$ values $W(i,j)$ satisfy $W(i,i)=0$ and $W(i,j)\ge 1$ for
$i\ne j$. Then every min-plus factorization
$W(i,j)=\min_t\bigl(u_t(i)+v_t(j)\bigr)$ has at least $N$ terms.
\end{lemma}

\begin{IEEEproof}
Since the minimum is attained everywhere, each diagonal cell $(i,i)$ has a
\emph{tight} term $t$ with $u_t(i)+v_t(i)=0$ (in particular both values are
finite), and every term obeys $u_t(i)+v_t(j)\ge W(i,j)$ for all $i,j$.
Suppose one term $t$ is tight at two diagonal cells $(i,i)$ and $(j,j)$
with $i\ne j$. Then
\begin{align*}
  0&=\bigl(u_t(i)+v_t(i)\bigr)+\bigl(u_t(j)+v_t(j)\bigr)\\
   &=\bigl(u_t(i)+v_t(j)\bigr)+\bigl(u_t(j)+v_t(i)\bigr)\\
   &\ge W(i,j)+W(j,i)\ \ge\ 2,
\end{align*}
a contradiction. So the $N$ diagonal cells require $N$ distinct tight
terms.
\end{IEEEproof}

\begin{IEEEproof}[Proof of Theorem~\ref{thm:main}]
Proposition~\ref{prop:doors} gives factorization rank at most $2^{s}$.
Lemmas~\ref{lem:block} and~\ref{lem:crossing} give factorization rank at
least $2^{s}$, since a factorization of $W$ restricts to a factorization of
any submatrix. For the tropical rank, the block of Lemma~\ref{lem:block} is
tropically nonsingular: the identity permutation contributes $0$ to the
tropical determinant, while any other permutation moves at least two
indices and hence contributes at least $2$; so the minimum is attained
uniquely. Therefore the tropical rank is at least $2^{s}$, and by the chain
tropical $\le$ Kapranov $\le$ Barvinok~\cite{DSS05} all three ranks are
squeezed to exactly $2^{s}$.
\end{IEEEproof}

\begin{remark}
The lower bound uses no structure of the factors beyond validity and
tightness: no layering, no sequential state evolution, no linearity. Thus
the minimal-trellis state count, classically optimal among trellis
realizations~\cite{Muder88,Vardy98}, is optimal among arbitrary min-plus
factorizations of the conditional value table.
\end{remark}

\begin{remark}
The proof is insensitive to the binary alphabet in all but notation: for a
linear code over $\F_q$ with the Hamming metric the same statements hold
with $q^{s}$ in place of $2^{s}$, by the same two lemmas. We state the
binary case for concreteness.
\end{remark}

\section{Graphs, Signings and Frustration}

Let $G=(V,E)$ be a connected graph with $n=|V|$ vertices, and let
$B_G\subseteq\F_2^{E}$ be its cut (cocycle) space, of dimension $n-1$: the
span of the vertex cuts $\delta(S)$, $S\subseteq V$. Identify a signing
$\sigma:E\to\{\pm1\}$ with the indicator $\omega\in\F_2^E$ of its negative
edges. A choice of spins $\varepsilon:V\to\{\pm1\}$ leaves exactly the edges
of $\omega\oplus\delta(S)$ frustrated, where $S=\varepsilon^{-1}(-1)$; hence
the \emph{frustration index} of the signing---the ground-state energy of the
Ising system with couplings $\sigma$---is the coset-leader weight
\[
  \ell(\omega)\;=\;\min_{S\subseteq V}\,|\omega\oplus\delta(S)|
  \;=\;d(\omega,B_G)\;=\;w_{B_G}(\omega).
\]
For an edge bipartition $E=E_L\sqcup E_R$, the conditional decoding matrix
of $B_G$ is the \emph{conditional ground-state matrix}: the table of
optimal energies as a function of the two halves of the coupling pattern.
Theorem~\ref{thm:main} applies verbatim, with
\[
  C_L=\{\delta(S)\subseteq E_L\},\qquad C_R=\{\delta(S)\subseteq E_R\},
\]
the cuts supported on one side.

A natural family of bipartitions makes the state complexity explicit.

\begin{corollary}\label{cor:spoked}
Let $V=A\sqcup B$ with $G[A]$ and $G[B]$ connected and every vertex of $B$
adjacent to some vertex of $A$, and set $E_L=E(A)\cup E(A,B)$,
$E_R=E(B)$. Then $s=|B|-1$, and consequently the conditional ground-state
matrix has min-plus factorization rank and tropical rank exactly
$2^{\,|B|-1}$.
\end{corollary}

\begin{IEEEproof}
\emph{$C_R=0$:} suppose $\delta(S)\subseteq E(B)$ with
$S\notin\{\emptyset,V\}$. No cut edge is in $E(A,B)$, so each $b\in B$ lies
on the same side as each of its neighbours in $A$; no cut edge lies in
$E(A)$ and $G[A]$ is connected, so all of $A$ lies on one side; hence every
$b\in B$ lies on $A$'s side, and $S$ is trivial---a contradiction.
\emph{$\dim C_L=|A|$:} if $\delta(S)\cap E(B)=\emptyset$ then, since $G[B]$
is connected, $S\cap B\in\{\emptyset,B\}$, so up to complementation
$S=T\subseteq A$; the map $T\mapsto\delta(T)$ on subsets of $A$ is linear
with trivial kernel (only $\delta(\emptyset)$ and $\delta(V)$ vanish, and
$A\ne V$), so its image has dimension $|A|$. Hence
$s=(n-1)-|A|-0=|B|-1$.
\end{IEEEproof}

Note that the connectivity of $G[A]$ is genuinely needed for $C_R=0$: if
$G[A]$ is disconnected one can have single-edge cuts inside $E(B)$ (this
occurs already on trees), and indeed the hypotheses force the cycle rank of
$G$ to be at least $|B|-1$, so trees satisfy them only degenerately.

\begin{example}[Petersen graph]\label{ex:petersen}
Let $G$ be the Petersen graph with $A$ the outer $5$-cycle and $B$ the
inner pentagram, all five spokes present. Then $n=10$, $\dim B_G=9$,
$\dim C_L=5$, $C_R=0$, and $s=|B|-1=4$: the conditional ground-state matrix
across the cut $E_L=\text{outer}\cup\text{spokes}$, $E_R=\text{inner}$ has
min-plus rank exactly $16$. The transversal block of Lemma~\ref{lem:block}
consists of sixteen full cuts, one per state class, with zero diagonal and
off-diagonal entries $\ge 1$. The complementary cut ($E_L=\text{outer}$,
$E_R=\text{spokes}\cup\text{inner}$) has $\dim C_L=0$, $\dim C_R=5$ and the
same $s=4$. For orientation: $B_G$ has $2^{6}=64$ cosets in $\F_2^{15}$,
with coset-leader weight distribution $\{0{:}1,\ 1{:}15,\ 2{:}45,\ 3{:}3\}$,
so frustration indices range over $0,\dots,3$.
\end{example}

\begin{remark}[computational verification]\label{rem:verify}
The block construction, the dimension counts, and the
diagonal/off-diagonal pattern of Lemma~\ref{lem:block} were verified
exhaustively by computer for: the Petersen graph at both cuts above
(including the coset census of Example~\ref{ex:petersen}); the planar prism
$C_4\times K_2$; and random connected graphs on $7$ vertices with
\emph{random, unstructured} edge bipartitions, confirming the general
Theorem~\ref{thm:main} beyond the structured hypotheses of
Corollary~\ref{cor:spoked}. Code is available from the author on request.
\end{remark}

\section{Local Triviality Versus Global Rank}

On high-girth graphs the exponential rank of Corollary~\ref{cor:spoked}
coexists with complete local featurelessness of the signings themselves.
Recall that two signings are \emph{switching-equivalent} if they differ by
a cut $\delta(S)$ (Harary~\cite{Harary53}; Zaslavsky~\cite{Zaslavsky82});
switching-invariant quantities are exactly the functions of the class
$\omega\oplus B_G$, and $\ell$ is one such.

\begin{proposition}\label{prop:local}
Let $G$ have girth $g$, let $v\in V$, and let $r\ge 0$ satisfy $2r+1<g$.
Then the subgraph induced on the ball $B_r(v)$ is a tree, and consequently
every signing of $G$, restricted to $B_r(v)$, is switching-equivalent to
the all-positive signing by switchings supported inside the ball.
\end{proposition}

\begin{IEEEproof}
If $xy$ were an edge inside $B_r(v)$ other than an edge of a breadth-first
search tree from $v$, the two tree paths from $v$ to $x$ and $y$ together
with $xy$ would contain a cycle of length at most $2r+1<g$, a
contradiction; so the induced ball is the BFS tree. On a tree, process
vertices outward from the root: switching a vertex flips the sign of its
parent edge, so all edge signs can be set positive. (Equivalently: forests
are balanced, and balanced signings are
switching-trivial~\cite{Harary53,Zaslavsky82}.)
\end{IEEEproof}

\begin{corollary}\label{cor:blind}
Fix $r$ with $2r+1<g$. Every statistic of a signing that is a function of
the switching classes of its radius-$r$ balls is constant across all
$2^{|E|}$ signings of $G$. In particular, for the families of
Corollary~\ref{cor:spoked} with girth $g>2r+1$, all radius-$r$ local
information about the coupling pattern is void, while the conditional
ground-state matrix across the cut has the maximum possible min-plus rank
$2^{\,|B|-1}$ permitted by its state complexity.
\end{corollary}

The pairing quantifies a familiar physical intuition: for frustration-type
objectives on high-girth graphs, conditional values are globally maximally
incompressible in the min-plus model even though locally there is,
provably, nothing to see.

\section{Concluding Remarks}

\begin{remark}[rank is not hardness]\label{rem:planar}
The prisms $C_k\times K_2$ satisfy the hypotheses of
Corollary~\ref{cor:spoked} with $|B|=k$, so their conditional ground-state
matrices have min-plus rank $2^{\,k-1}$---yet these graphs are planar, and
ground states of planar Ising systems are computable in polynomial
time~\cite{Barahona82}. Exponential conditional rank is therefore fully
compatible with polynomial-time solvability: the rank measures how
compressible the \emph{table of conditional optimal values} is in the
min-plus model, i.e., the unavoidable width of any min-plus message across
the cut, and a polynomial-time algorithm is free to compute with the
instance by other means (for planar Ising, via $T$-joins and matchings).
Conversely, computing $\ell$ on general graphs is
NP-hard~\cite{Barahona82}, while bounded factor rank of the relevant
matrices makes related optimization problems efficient~\cite{Shitov14}. We
flag this explicitly so that Theorem~\ref{thm:main} is not mistaken for a
lower bound on computation.
\end{remark}

\begin{remark}[an exactly solvable family for tropical rank]
Deciding whether the Barvinok rank of a matrix is at most a fixed constant
is NP-hard~\cite{Shitov14}, and exact evaluations on natural structured
families are correspondingly scarce. Theorem~\ref{thm:main} provides an
infinite family of $\{0,1,2,\dots\}$-valued matrices---doubly exponential
in size, described by polynomial-size data $(C,L,R)$---whose Barvinok,
tropical and Kapranov ranks are all known exactly and coincide.
\end{remark}

\begin{remark}[rank profiles]
Running the cut through all positions $i=0,\dots,m$ of a coordinate
ordering turns Theorem~\ref{thm:main} into the statement that the min-plus
rank profile of the conditional decoding matrices equals the
state-complexity profile of the code. All results on dimension/length
profiles and on optimizing state complexity over coordinate
permutations~\cite{Forney94,KTFL93,Vardy98} therefore transfer verbatim to
statements about tropical ranks.
\end{remark}

\begin{remark}[open questions]
(i) Group codes and codes over rings: the crossing argument is
alphabet-free, but the fiber invariance uses linearity; how far does the
identity extend? (ii) Approximate factorization: how does the min-plus rank
of $W$ degrade if entries may be computed within additive error
$\varepsilon$, and is there a robust analogue of the state complexity?
(iii) Beyond the Hamming metric: for which weight functions on $\F_2^m$
does conditional-value rank still equal a state-space count?
\end{remark}

\section*{Acknowledgment}
This work was developed with the assistance of large language models
(Anthropic's Claude), used for derivation, drafting, literature search, and
computational verification. The mathematical content and its correctness
remain the sole responsibility of the author.

\end{document}